\documentclass{article}

\usepackage{fullpage}
\usepackage{setspace}
\usepackage{authblk}
\usepackage{apacite}

\usepackage[latin9]{inputenc}
\usepackage[english]{babel}
\usepackage{amsmath}
\usepackage{amssymb}
\usepackage{graphicx}

\makeatletter

\usepackage{color}
\usepackage{bm}
\usepackage{bbm}
\usepackage{nicefrac}
\usepackage{enumerate}
\usepackage{amsthm}
\usepackage{todonotes}

\renewcommand{\v}[1]{\bm{#1}}
\newcommand{\E}{\mathbb{E}}
\newcommand{\Q}{\mathbb{Q}}
\newcommand{\eq}[1]{equation (\ref{eq:#1})}

\newcommand{\fig}[1]{figure \ref{fig:#1}}
\newcommand{\Fig}[1]{Figure \ref{fig:#1}}
\DeclareMathOperator{\diag}{diag}
\DeclareMathOperator{\Tr}{Tr}
\newcommand{\norm}[1]{\left\lVert#1\right\rVert_{1}}
\renewcommand{\d}[1]{\,d#1}
\newtheorem{assumption}{Assumption}\newtheorem{theorem}{Theorem}\newtheorem{corollary}{Corollary}\newtheorem{proposition}{Proposition}\newtheorem{lemma}{Lemma}

\newcommand{\Exq}{\mathbb{E}^Q} 
\newcommand{\Rss}{\Xi_{ss}}
\newcommand{\Rsd}{\Xi_{sd}}
\newcommand{\Rds}{\Xi_{ds}}
\newcommand{\Rdd}{\Xi_{dd}}
\newcommand{\pss}{\pi_{ss}}
\newcommand{\psd}{\pi_{sd}}
\newcommand{\pds}{\pi_{ds}}

\newcommand{\Matrix}[2]{ \left( \begin{array}{#1} #2 \end{array} \right) }

\newcommand{\Ind}{\mathbbm{1}}

\makeatother

\begin{document}
\title{Cross-ownership as a structural explanation for rising correlations
in crisis times}
\author{Nils Bertschinger}
\author[2]{Axel A. Araneda}
\affil[1]{Frankfurt Institute for Advanced Studies; D-60438 Frankfurt am Main, Germany.}
\affil[2]{Institute of Financial Complex Systems, Faculty of Economics and Administration, Masaryk University; 602 00 Brno, Czech Republic.}

\maketitle

\begin{abstract}
In this paper, we examine the interlinkages among firms through a
financial network where cross-holdings on both equity and debt are
allowed. We relate mathematically the correlation among equities with
the unconditional correlation of the assets, the values of their business
assets and the sensitivity of the network, particularly the $\Delta$-Greek.
We noticed also this relation is independent of the Equities level.
Besides, for the two-firms case, we analytically demonstrate that
the equities correlation is always higher than the correlation of
the assets; showing this issue by numerical illustrations. Finally,
we study the relation between equity correlations and asset prices,
where the model arrives to an increase in the former due to a fall
in the assets.
\end{abstract}

\section{Introduction}

The famous \citeA{merton1974pricing} model relates debt and equity of
a firm with European put and call options respectively. Since then it
has been developed into an industry standard for structural credit
risk modeling and management.  However, the increasingly complex
interlinkages between financial institutions are at odds with an
individual and separate valuation of risk (see \citeA{de2000systemic}
for an early survey of systemic risk). Especially, the latest
financial crisis has painfully revealed the danger of contagion
throughout the financial system and spurred a wealth of interest in
theoretical models of systemic risk. In this regard, the model of
\citeA{eisenberg2001systemic} and its ``clearing payment vector''
insight have been identified as the seminal contribution in the field,
forming the basis for numerous studies of financial contagion arising
from cross-ownership of debt \cite{cifuentes2005liquidity,gai2010contagion,elliott2014financial}.
See \citeA{caccioli2018network} and \citeA{sasidevan2019systemic}
for recent surveys of the network approach into systemic risk.

An interesting and alternative viewpoint is provided by
\citeA{suzuki2002valuing}. While the model can be interpreted as an
extension of the \citeA{eisenberg2001systemic} model, it is better
seen as an extension of the Merton model allowing for multiple firms
with cross-ownerwhip of debt as well as equity. Thereby, considering
financial contagion as a problem of firm valuation where debt and
equity have to be assessed in a self-consistent fashion, e.g.~solving
a fixed point via Picard iteration 	\cite{Hain2015}. Furthermore, Suzuki explicitly solves the
valuation problem in case of two financial institutions	,
conditional on the values of the business assets where banks are
solvent or in default (Suzuki areas). For three or more banks a formal solution can still be written down, but requires a case distinction between exponentially many solvency
regions and can no longer be visualized in two dimensions.  Further developments on the
\citeauthor{suzuki2002valuing} model extend it to debts of multiple
seniorities \cite{fischer2014no}, address the (joint) default
probabilities under this model \cite{karl2014cross} or compute
analytic bounds assuming comonotonic asset endowments
\cite{feinstein2021}.

On the empirical side, it is an established ``stylized fact'' that
correlations rise during bear markets and crises times
\cite{longin2001extreme,ang2002asymmetric,kalkbrener2015correlation}.
Indeed, \citeA{Baig1999} found that during the Asian crisis,
correlations in stock markets, interest rates, exchange rates and
sovereign spreads rose significantly as compared to tranquil
times. \citeA{onnela2003dynamics} investigated the empirical
distribution of pairwise stock correlation coefficients of stocks
traded at NYSE, estimating correlations in a rolling window fashion, finding a substantial increase of their mean value around the
Black Monday of Oct 1987.  \citeA{preis2012quantifying} analyzed, with
the time-varying correlations among the 30 stocks composing the DJIA
index and demonstrate a linear relationship with market stress.
Instead, \citeA{adams2017correlations} argue, based on econometric
consideration, that correlations change in a step-like fashion due to
particular financial events (structural breaks). Similar arguments are
put forward line in more recent works of \citeA{choi2019self} and
\citeA{demetrescu2019testing}.

In terms of modeling, the works of
\citeA{cizeau2001correlation,lillo2000symmetry,kyle2001contagion} have
addressed the correlation issue from a structural perspective. Another
interesting approach in this line is provided by
\citeA{cont2013running} who evaluate how fire sales lead to
\emph{endogoneous} correlations in a simple multi-period
model,. However, to our knowledge, it has not been approached in the
context of cross-holding networks. Here, we show that network models
readily explain the correlation stylized fact. In particular, we
depart from a simply specification of the
\citeauthor{suzuki2002valuing} model, and proof that it exhibits
structural changes in the correlation between firm equity values
depending on solvency conditions. Furthermore, we seek to understand
and quantify the precise influence of different model parameters on
the observed correlation structure.

The remainder of the paper is structured as follows: First, we lay out
the general network valuation model following
\citeA{suzuki2002valuing}. Then, we interpret
values of firm equity and debt as derivative contracts, i.e.~extending
the \citeauthor{merton1974pricing} model to multiple firms, and link
the correlation between these derivatives with the $\Delta$
sensitivities, the asset prices, and the leverage. In turn,
concentrating on the two-firms case, we proof that the correlation
among the two firm equities never falls below their unconditional
asset correlation. Finally, we illustrate our results through
numerical simulations, showing in particular that correlations tend to
rise when equity prices drop. Thereby, providing a novel structural
explanation of this well-known stylized fact.

\section{Model}

\label{sec:Model}

\subsection{Notation and mathematical preliminaries}

Here we quickly summarize the mathematical notation employed in this
paper. We write vectors $\v{x},\v{y}\in\mathbb{R}^{n}$ with bold
lower case and matrices $\v{A},\v{B}\in\mathbb{R}^{m\times n}$ with
bold upper case letters. Individual entries of vectors and matrices
are written as $x_{i},A_{ij}$. $\diag(\v{x})$ denotes the $n\times n$
diagonal matrix $\v{D}$ with entries $D_{ii}=x_{i}$ along its diagonal.
The transpose of a matrix is denoted as $\v{A}^{T}$. All products
containing vectors and matrices are understood as standard matrix
products, e.g. $\v{A}\v{B}$ denotes the matrix product of $\v{A}$
and $\v{B}$, $\v{x}\v{x}$ is undefined whereas $\v{x}^{T}\v{x}$
is the scalar product of $\v{x}$ with itself. Row- and column-wise
stacking of vectors or matrices is denoted by $(\v{x};\v{y})$ and
$(\v{x},\v{y})$ respectively, i.e. $(\v{x};\v{y})$ is a $2n$-dimensional
vector whereas $(\v{x},\v{y})$ is a $n\times2$ matrix.

Random variables $X,Y$ are written as upper case letters with
individual outcomes $x,y$ denoted in lower case. Expectations are
denoted as $\E[f(X)]$ and understood with respect to the (joint)
distribution of random variables within the brackets. Sometimes we use
$\E_{t}^{\Q}$ to denote that the expectation is taken over the
risk-neutral measure $\Q$, implicitly conditioned on the information
filtration $\mathcal{F}_t$ up to time $t$.

\subsection{Network valuation}

\citeA{merton1974pricing} has shown that equity and firm debt can
be considered as call and put options on the firm's value respectively.
In this model, a single firm is holding externally priced assets $a$
and zero-coupon debt with nominal amount $d$ due at a single, fixed
maturity $T$. Then, at time $T$ the value of equity $s$ and the
recovery value of debt $r$ are given as 
\begin{align}
s & =\max\{0,a-d\}=(a-d)^{+},\label{eq:Merton}\\
r & =\min\{d,a\}=d-(d-a)^{+}
\end{align}
corresponding to an implicit call and put option respectively.

\citeA{suzuki2002valuing} and others \cite{elsinger2009financial,fischer2014no}
have since generalized this model to multiple firms with equity and
debt cross-holdings. In this paper we consider $n$ firms. Each firm
$i=1,\ldots,n$ holds an external asset $a_{i}>0$ as well as a fraction
$M_{ij}^{s}$ of firm $j$'s equity and debt $M_{ij}^{d}$. Here,
the investment fractions $M_{ij}^{s}$ and $M_{ij}^{d}$ are bounded
between $0$ and $1$, i.e. $0\leq M_{ij}^{s,d}\leq1$, and the actual
value invested in the equity of counterparty $j$ is given as $M_{ij}^{s}s_{j}$.
In the following we require: \begin{assumption} \label{assu_1} There
are no self-holdings, i.e. $M_{ii}^{s}=M_{ii}^{d}=0$ for all $i=1,\ldots,n$,
nor short positions, i.e. $M_{ij}^{s},M_{ij}^{d}\geq0$ for all $i,j=1,\ldots,n$.
Moreover, we require that the total fractions equity and debt held
by any counterparty cannot exceed unity. In addition, we assume that
some of each firms equity and debt are held externally, i.e. for all
$j=1,\ldots,n$ it holds that 
\begin{align}
\sum_{i}M_{ij}^{s}<1\quad\mbox{and}\quad\sum_{i}M_{ij}^{d}<1\,.\label{eq:assu_1a}
\end{align}
\end{assumption} That is, $\v{M}^{s}$ and $\v{M}^{d}$ are strictly
(left) sub-stochastic matrices. Alternatively, we can express this
as $\norm{\v{M}^{d}},\norm{\v{M}^{s}}<1$.

Now, the value of all assets $v_{i}$ held by firm $i$ is given by
\begin{align}
v_{i} & =a_{i}+\sum_{j=1}^{n}M_{ij}^{s}s_{j}+\sum_{j=1}^{n}M_{ij}^{d}r_{j}\,.\label{eq:XOS_val_i}
\end{align}
Correspondingly, the firm's equity and recovery value of debt are
given by 
\begin{align}
s_{i} & =\max\left\{ 0,a_{i}+\sum_{j}M_{ij}^{s}s_{j}+\sum_{j}M_{ij}^{d}r_{j}-d_{i}\right\} ,\label{eq:XOS_sr_i}\\
r_{i} & =\min\left\{ d_{i},a_{i}+\sum_{j}M_{ij}^{s}s_{j}+M_{ij}^{d}r_{j}\right\} \,.\label{eq:XOS_sr_i_2}
\end{align}
In matrix notation, i.e. collecting equity and debt values into vectors
$\v{s}=(s_{1},\ldots,s_{n})^{T}$ and $\v{r}=(r_{1},\ldots,r_{n})^{T}$
respectively, this can be rewritten as 
\begin{align}
\v{s} & =\max\left\{ \v{0},\v{a}+\v{M}^{s}\v{s}+\v{M}^{d}\v{r}-\v{d}\right\} ,\label{eq:XOS_vec}\\
\v{r} & =\min\left\{ \v{d},\v{a}+\v{M}^{s}\v{s}+\v{M}^{d}\v{r}\right\} 
\end{align}
Thus, the firms' equity and debt values are endogenously defined as
the solution of a fixed point. This is readily seen when collecting
equity and debt row-wise into a single vector $\v{x}=(\v{s};\v{r})$,
i.e. $\v{s}=\v{x}_{1:n}$ and $\v{r}=\v{x}_{(n+1):2n}$, and writing
\begin{align}
\v{x}=\v{g}(\v{a},\v{x})\label{eq:XOS_fix}
\end{align}
with the vector valued function $\v{g}=(g_{1}^{s},\ldots,g_{n}^{s},g_{1}^{r},\ldots,g_{n}^{r})^{T}$
where for $i=1,\ldots,n$ 
\begin{align}
g_{i}^{s}(\v{a},\v{x}) & =\max\left\{ 0,a_{i}+\sum_{j}M_{ij}^{s}x_{j}+\sum_{j}M_{ij}^{d}x_{n+j}-d_{i}\right\} ,\label{eq:XOS_g}\\
g_{i}^{r}(\v{a},\v{x}) & =\min\left\{ d_{i},a_{i}+\sum_{j}M_{ij}^{s}x_{j}+\sum_{j}M_{ij}^{d}x_{n+j}\right\} \,.\\
\end{align}
Each of the functions $g_{i}^{s}$ and $g_{i}^{r}$ is continuous
and increasing in $\v{a}$ and $\v{x}$. Together with assumption
\ref{assu_1} it follows that the fixed point of (\ref{eq:XOS_fix})
is positive and unique. \begin{theorem} Suppose that assumption \ref{assu_1}
holds. Then, for each value of external assets $\v{a}>\v{0}$ there
is a positive and unique $\v{x}$ solving (\ref{eq:XOS_fix}). \end{theorem}
\begin{proof} Our model is a special case of the one considered by
\cite{fischer2014no} with $k=1$ and $\v{d}_{\v{r}^{1},\v{r}^{0}}^{1}\equiv\v{d}$.
Furthermore, Fischer's assumption 3.1 holds by assumption \ref{assu_1}
and assumptions 3.6 and 3.7 are trivial as our nominal debt vector
$\v{d}$ is constant. The result then follows by his theorem 3.8 (iv).
\end{proof}

\section{Risk-neutral valuation}

The celebrated Merton model exploits the connection of Eq. \ref{eq:XOS_fix}
with option prices to obtain the ex-ante market prices at time $t<T$
as 
\begin{align}
s_{t} & =\E_{t}^{Q}[e^{-r\tau}S_{T}]=\E_{t}^{Q}[e^{-r\tau}(A_{T}-d)^{+}] &r_{t} =\E_{t}^{Q}[e^{-r\tau}R_{T}]=\E_{t}^{Q}[e^{-r\tau}(d-(d-A_{T})^{+})]\label{eq:MertonQ}
\end{align}
respectively. Furthermore, assuming a geometric Brownian motion for
the price of the external assets, i.e. 
\begin{align}
dA_{t} & =rA_{t}\d{t}+\sigma_{a}A_{t}\d{W_{t}^{Q}}\label{eq:GB1D}
\end{align}
the corresponding stochastic differential equation for $s_{t}$ can
be obtained via Ito's lemma as 
\begin{align}
dS_{t} & =\left(\frac{\partial s_{t}}{\partial t}+\frac{\partial s_{t}}{\partial a_{t}}rA_{t}+\frac{1}{2}\frac{\partial^{2}s_{t}}{\partial a_{t}^{2}}\sigma_{a}^{2}A_{t}^{2}\right)\d{t}+\frac{\partial s_{t}}{\partial a_{t}}\sigma_{a}A_{t}\d{W_{t}^{Q}}\;.\label{eq:ItoVola}
\end{align}
Matching the volatility with $\sigma_{s}S_{t}$ one obtains the well
known relation 
\begin{align}
\sigma_{s} & =\frac{\partial s_{t}}{\partial a_{t}}\sigma_{a}\frac{a_{t}}{s_{t}}=\sigma_{a}\Delta\lambda\label{eq:MertonVola}
\end{align}
between equity and asset volatility. Here, $\Delta=\frac{\partial s_{t}}{\partial a_{t}}$
is the option Delta and $\lambda=\frac{A_{t}}{S_{t}}$ its leverage.

\subsection{Network valuation}

Denoting the unique solution of \eq{XOS_fix} by $\v{x}^{*}(\v{a})$,
we can consider the corresponding value of equity and debt claims
as derivative contracts on the underlying $\v{a}$. Accordingly, the
ex-ante market price at time $t<T$ is given as 
\begin{align}
\v{x}_{t}=\E_{t}^{Q}[e^{-r\tau}\v{x}^{*}(\v{A}_{T})]\label{eq:val_Q}
\end{align}
with the risk-less interest rate $r$ and time to maturity $\tau=T-t$.
The expectation is taken with respect to the risk-neutral measure
$Q$ of external asset values $\v{a}$ at maturity $T$. In the following,
we assume that the risk-neutral asset values follow a multi-variate
geometric Brownian motion, i.e. 
\begin{align}
d\v{A}_{t} & =r\v{A}_{t}\d{t}+\diag(\v{\sigma})\diag(\v{A}_{t})\d{\v{W}_{t}^{Q}}\label{eq:SDE}
\end{align}
with possibly correlated Wiener processes $\v{W}_{t}^{Q}$, i.e. $\E[\d{W_{i,t}^{Q}}\d{W_{j,t}^{Q}}]=\rho_{ij}\d{t}$
with $\rho_{ii}=1$.

The well-known solution of \eq{SDE} is given by 
\begin{align}
\v{A}_{t} & =\v{a}_{0}e^{\left(r-\frac{1}{2}\diag(\v{\sigma}^{2})\right)t+\diag(\v{\sigma})\v{W}_{t}}\label{eq:p_A_t}
\end{align}
where $\v{a}_{0}>0$ denotes the initial value and $\v{W}_{t}$ is
multivariate normal distributed with mean $\v{0}$ and covariance
matrix $t\v{C}$ with entries $C_{ij}=\rho_{ij}$.

As before, via the multi-variate Ito formula we obtain 
\begin{align}
dX_{i,t} & =\left(\frac{\partial x_{i,t}}{\partial t}+r\left(\frac{\partial x_{i,t}}{\partial\v{a}_{t}}\right)^{T}\v{A}_{t}+\frac{1}{2}\Tr\left(\diag(\v{A}_{t})^{T}\diag(\v{\sigma})^{T}\frac{\partial^{2}x_{i,t}}{\partial\v{a}_{t}\partial\v{a}_{t}}\diag(\v{\sigma})\diag(\v{A}_{t})\right)\right)\d{t}\nonumber \\
 & +\left(\frac{\partial x_{i,t}}{\partial\v{a}_{t}}\right)^{T}\diag(\v{\sigma})\diag(\v{A}_{t})\d{\v{W}_{t}^{Q}}\;.
\end{align}
Then, again matching the volatility with $\v{\sigma}_{i,x}x_{i,t}$
we compute the equity and debt volatilities 
\begin{align}
\v{\sigma}_{i,x} & =\frac{1}{x_{i,t}}\left(\frac{\partial x_{i,t}}{\partial\v{a}_{t}}\right)^{T}\diag(\v{\sigma})\diag(\v{a}_{t})\label{eq:VolaNet}
\end{align}
or collecting all terms into a volatility matrix 
\begin{align}
\v{L}_{x} & =\diag(\v{x}_{t})^{-1}\left(\frac{\partial\v{x}_{t}}{\partial\v{a}_{t}}\right)^{T}\diag(\v{\sigma})\diag(\v{a}_{t})\;.\label{eq:VolaNetMatrix}
\end{align}
Note that the instantaneous covariance matrix of $\v{x}_{t}$ at time
$t$ is then given by 
\begin{align}
\v{\Sigma}_{x} & =\v{L}_{x}t\v{C}\v{L}_{x}^{T}\label{eq:SigmaNet}\\
 & =\diag(\v{x}_{t})^{-1}\left(\frac{\partial\v{x}_{t}}{\partial\v{a}_{t}}\right)^{T}\diag(\v{\sigma})\diag(\v{a}_{t})t\v{C}\diag(\v{a}_{t})\diag(\v{\sigma})\left(\frac{\partial\v{x}_{t}}{\partial\v{a}_{t}}\right)\diag(\v{x}_{t})^{-1}\\
 & =\diag(\v{x}_{t})^{-1}\left(\frac{\partial\v{x}_{t}}{\partial\v{a}_{t}}\right)^{T}diag(\v{a}_{t})\v{\Sigma}_{a}\diag(\v{a}_{t})\left(\frac{\partial\v{x}_{t}}{\partial\v{a}_{t}}\right)\diag(\v{x}_{t})^{-1}
\end{align}
where $\v{\Sigma}_{a}=\diag(\v{\sigma})t\v{C}\diag(\v{\sigma})^{T}$
denotes the instantaneous asset covariance at time $t$. This generalizes
\eq{MertonVola} with the Delta matrix $\v{\Delta}_{x}=\frac{\partial\v{x}_{t}}{\partial\v{a}_{t}}$
and $\diag(\v{x}_{t})^{-1},\diag(\v{a}_{t})$ acting as leverage.
In contrast, to the uni-variate case these cannot be collected into
a leverage matrix as $\diag(\v{x}_{t})^{-1}$ is multiplied from the
left, i.e. acts on the rows, whereas $\diag(\v{a}_{t})$ is multiplied
from the right, acts on the columns.

\section{Two bank case}

For two banks, i.e. $i=1,2$, the fixed point equations \ref{eq:XOS_g}
for equity and debt can be solved explicitly. In particular, \citeA{suzuki2002valuing}
has shown that the value of the equity and debt, at maturity, depends
on the solvency conditions of the firms. Here, firm $i$ is solvent
(insolvent) if its total value \ref{eq:XOS_val_i} exceeds (falls
short of) its nominal debt, i.e. $v_{i,T}\geq(<)d_{i}$. Following
Suzuki, we define four regions (Suzuki areas) which consider the combinatory
of solvency or default condition at maturity \citeA{suzuki2002valuing,karl2014cross}:

\[
\Rss=\left\{ \left(a_{1,T},a_{2,T}\right)\in\mathbb{R}_{+}^{2}:v_{1,T}\geq d_{1}\land v_{2,T}\geq d_{2}\right\} 
\]

\[
\Rsd=\left\{ \left(a_{1,T},a_{2,T}\right)\in\mathbb{R}_{+}^{2}:v_{1,T}\geq d_{1}\land v_{2,T}<d_{2}\right\} 
\]

\[
\Rds=\left\{ \left(a_{1,T},a_{2,T}\right)\in\mathbb{R}_{+}^{2}:v_{1,T}<d_{1}\land v_{2,T}\geq d_{2}\right\} 
\]

\[
\Rdd=\left\{ \left(a_{1,T},a_{2,T}\right)\in\mathbb{R}_{+}^{2}:v_{1,T}<d_{1}\land v_{2,T}<d_{2}\right\} 
\]
\\

After that, and using simply assumptions \citeA{suzuki2002valuing},
we have a fix-point-solution for the system \ref{eq:XOS_sr_i}-\ref{eq:XOS_sr_i_2}
conditional to each Suzuki area, given by:

\begin{equation}
s_{1,T}=\begin{cases}
\frac{a_{1,T}-d_{1}+M_{12}^{r}d_{2}+M_{12}^{s}\left(a_{2,T}-d_{2}+M_{21}^{r}d_{1}\right)}{1-M_{12}^{s}M_{21}^{s}} & ,\,\left(a_{1,T},a_{2,T}\right)\in\Rss\\
\frac{a_{1,T}-d_{1}+M_{12}^{r}d_{2}+M_{12}^{r}\left(a_{2,T}-d_{2}+M_{21}^{r}d_{1}\right)}{1-M_{12}^{r}M_{21}^{s}} & ,\,\left(a_{1,T},a_{2,T}\right)\in\Rsd\\
0 & ,\,\left(a_{1,T},a_{2,T}\right)\in\Rds\\
0 & ,\,\left(a_{1,T},a_{2,T}\right)\in\Rdd
\end{cases}\label{eq:S1}
\end{equation}

\begin{equation}
s_{2,T}=\begin{cases}
\frac{a_{2,T}-d_{2}+M_{21}^{r}d_{1}+M_{21}^{s}\left(a_{1,T}-d_{1}+M_{12}^{r}d_{2}\right)}{1-M_{12}^{s}M_{21}^{s}} & ,\,\left(a_{1,T},a_{2,T}\right)\in\Rss\\
0 & ,\,\left(a_{1,T},a_{2,T}\right)\in\Rsd\\
\frac{a_{2,T}-d_{2}+M_{21}^{r}d_{1}+M_{21}^{r}\left(a_{1,T}-d_{1}+M_{12}^{r}d_{2}\right)}{1-M_{12}^{s}M_{21}^{d}} & ,\,\left(a_{1,T},a_{2,T}\right)\in\Rds\\
0 & ,\,\left(a_{1,T},a_{2,T}\right)\in\Rdd
\end{cases}\label{eq:S2}
\end{equation}

\begin{equation}
r_{1,T}=\begin{cases}
d_{1} & ,\,\left(a_{1,T},a_{2,T}\right)\in\Rss\\
d_{1} & ,\,\left(a_{1,T},a_{2,T}\right)\in\Rsd\\
\frac{a_{1,T}+M_{12}^{r}d_{2}+M_{12}^{s}\left(a_{2,T}-d_{2}\right)}{1-M_{12}^{s}M_{21}^{r}} & ,\,\left(a_{1,T},a_{2,T}\right)\in\Rds\\
\frac{a_{1,T}+M_{12}^{r}a_{2,T}}{1-M_{12}^{r}M_{21}^{r}} & ,\,\left(a_{1,T},a_{2,T}\right)\in\Rdd
\end{cases}\label{eq:r1}
\end{equation}

\begin{equation}
r_{2,T}=\begin{cases}
d_{2} & ,\,\left(a_{1,T},a_{2,T}\right)\in\Rss\\
\frac{a_{2,T}+M_{21}^{r}d_{1}+M_{21}^{s}\left(a_{1,T}-d_{1}\right)}{1-M_{21}^{s}M_{12}^{r}} & ,\,\left(a_{1,T},a_{2,T}\right)\in\Rsd\\
d_{2} & ,\,\left(a_{1,T},a_{2,T}\right)\in\Rds\\
\frac{a_{2,T}+M_{21}^{r}a_{1,T}}{1-M_{12}^{r}M_{21}^{r}} & ,\,\left(a_{1,T},a_{2,T}\right)\in\Rdd
\end{cases}\label{r2}
\end{equation}

\subsection{Computing correlations}

Consider two assets with covariance matrix 
\[
\Sigma=\Matrix{cc}{\sigma_1^2 & \sigma_1 \sigma_2 \rho \\ \sigma_1 \sigma_2 \rho & \sigma_2^2} \;,
\]
i.e. with volatilities $\sigma_{1},\sigma_{2}$ and correlation $\rho$.

The Cholesky factor $\v{L}$ with $\Sigma=\v{L}\v{L}^{T}$ is then
given by 
\begin{equation}
    \v{L} = \Matrix{cc}{\sigma_1 & 0 \\ \sigma_2 \rho & \sigma_2 \sqrt{1
      - \rho^2}} = \Matrix{cc}{l_{11} & 0 \\ l_{21} & l_{22}} \; .
    \label{eq:cholesky}
\end{equation}
In particular, the correlation coefficient can be expressed directly
in terms of the Cholesky coefficients as 
\begin{equation}
\rho=\frac{l_{21}}{\sqrt{l_{21}^{2}+l_{22}^{2}}}=\frac{1}{\sqrt{1+\left(\frac{l_{22}}{l_{21}}\right)^{2}}}\;.\label{eq:rho_2}
\end{equation}

Furthermore, for the factor of the equity covariances we obtain 
\[
L^{s}=\mathrm{diag}(\v{s})^{-1}\frac{\partial\v{s}}{\partial\v{a}}\mathrm{diag}(\v{a})\v{L}
\]
or explicitly 
\begin{equation}
l_{ij}^{s}=\sum_{k}\frac{\Delta_{ik}a_{k}}{s_{i}}l_{kj}\label{eq:Lexplizit}
\end{equation}
where $\Delta_{ij}=\frac{\partial s_{i}}{\partial a_{j}}$. Note that
this is not a Cholesky factor, as it will not be lower triangular
in general. Nevertheless, we have $\Sigma^{s}=\v{L}^{s}(\v{L}^{s})^{T}$
or explicitly 
\begin{align*}
\sigma_{ij}^{s} & =\sum_{k}l_{ik}^{s}l_{jk}^{s}
\end{align*}
and therefore for the correlation coefficient 
\begin{align*}
\rho^{s} & =\frac{\sigma_{12}^{s}}{\sqrt{\sigma_{11}^{s}\sigma_{22}^{s}}}\\
 & =\frac{l_{11}^{s}l_{21}^{s}+l_{12}^{s}l_{22}^{s}}{\sqrt{((l_{11}^{s})^{2}+(l_{12}^{s})^{2})((l_{21}^{s})^{2}+(l_{22}^{s})^{2})}}\\
 & =\mathrm{sign}(l_{11}^{s}l_{21}^{s}+l_{12}^{s}l_{22}^{s})\sqrt{\frac{(l_{11}^{s}l_{21}^{s})^{2}+2l_{11}^{s}l_{21}^{s}l_{12}^{s}l_{22}^{s}+(l_{12}^{s}l_{22}^{s})^{2}}{(l_{11}^{s}l_{21}^{s})^{2}+(l_{11}^{s}l_{22}^{s})^{2}+(l_{12}^{s}l_{21}^{s})^{2}+(l_{12}^{s}l_{22}^{s})^{2}}}\;.
\end{align*}
Assuming that $l_{11}^{s}l_{21}^{s}+l_{12}^{s}l_{22}^{s}$ is nonzero,
we can rewrite the above equation using a quadratic extension as 
\begin{align}
\rho^{s} & =\mathrm{sign}(l_{11}^{s}l_{21}^{s}+l_{12}^{s}l_{22}^{s})\sqrt{\frac{(l_{11}^{s}l_{21}^{s})^{2}+2l_{11}^{s}l_{21}^{s}l_{12}^{s}l_{22}^{s}+(l_{12}^{s}l_{22}^{s})^{2}}{(l_{11}^{s}l_{21}^{s})^{2}+2l_{11}^{s}l_{21}^{s}l_{12}^{s}l_{22}^{s}-2l_{11}^{s}l_{21}^{s}l_{12}^{s}l_{22}^{s}+(l_{11}^{s}l_{22}^{s})^{2}+(l_{12}^{s}l_{21}^{s})^{2}+(l_{12}^{s}l_{22}^{s})^{2}}}\nonumber \\
 & =\mathrm{sign}(l_{11}^{s}l_{21}^{s}+l_{12}^{s}l_{22}^{s})\sqrt{\frac{(l_{11}^{s}l_{21}^{s})^{2}+2l_{11}^{s}l_{21}^{s}l_{12}^{s}l_{22}^{s}+(l_{12}^{s}l_{22}^{s})^{2}}{(l_{11}^{s}l_{21}^{s})^{2}+2l_{11}^{s}l_{21}^{s}l_{12}^{s}l_{22}^{s}+(l_{11}^{s}l_{22}^{s}-l_{12}^{s}l_{21}^{s})^{2}+(l_{12}^{s}l_{22}^{s})^{2}}}\nonumber \\
 & =\mathrm{sign}(l_{11}^{s}l_{21}^{s}+l_{12}^{s}l_{22}^{s})\sqrt{\frac{1}{1+\left(\frac{l_{11}^{s}l_{22}^{s}-l_{12}^{s}l_{21}^{s}}{l_{11}^{s}l_{21}^{s}+l_{12}^{s}l_{22}^{s}}\right)^{2}}}\label{eq:rho_s}
\end{align}

To compute the correlation coefficient we start with \eq{Lexplizit}
and obtain 
\begin{alignat*}{2}
l_{11}^{s} & =\frac{1}{s_{1}}(\Delta_{11}a_{1}l_{11}+\Delta_{12}a_{2}l_{21}) &  & =\frac{1}{s_{1}}(\Delta_{11}a_{1}\sigma_{1}+\Delta_{12}a_{2}\sigma_{2}\rho)\\
l_{12}^{s} & =\frac{1}{s_{1}}(\Delta_{11}a_{1}l_{12}+\Delta_{12}a_{2}l_{22}) &  & =\frac{1}{s_{1}}\Delta_{12}a_{2}\sigma_{2}\sqrt{1-\rho^{2}}\\
l_{21}^{s} & =\frac{1}{s_{2}}(\Delta_{21}a_{1}l_{11}+\Delta_{22}a_{2}l_{21}) &  & =\frac{1}{s_{2}}(\Delta_{21}a_{1}\sigma_{1}+\Delta_{22}a_{2}\sigma_{2}\rho)\\
l_{22}^{s} & =\frac{1}{s_{2}}(\Delta_{21}a_{1}l_{12}+\Delta_{22}a_{2}l_{22}) &  & =\frac{1}{s_{2}}\Delta_{22}a_{2}\sigma_{2}\sqrt{1-\rho^{2}}\\
\end{alignat*}
where the simplified expressions follow from $l_{12}=0$.

Overall, we obtain for the relevant terms 
\begin{align}
l_{11}^{s}l_{22}^{s} & =\frac{1}{s_{1}s_{2}}(\Delta_{11}a_{1}\sigma_{1}+\Delta_{12}a_{2}\sigma_{2}\rho)\Delta_{22}a_{2}\sigma_{2}\sqrt{1-\rho^{2}}\nonumber \\
 & =\frac{1}{s_{1}s_{2}}(\Delta_{11}a_{1}\Delta_{22}a_{2}\sigma_{1}\sigma_{2}\sqrt{1-\rho^{2}}+\Delta_{12}a_{2}\Delta_{22}a_{2}\sigma_{2}^{2}\rho\sqrt{1-\rho^{2}})\nonumber \\
l_{12}^{s}l_{21}^{s} & =\frac{1}{s_{1}s_{2}}\Delta_{12}a_{2}\sigma_{2}\sqrt{1-\rho^{2}}(\Delta_{21}a_{1}\sigma_{1}+\Delta_{22}a_{2}\sigma_{2}\rho)\nonumber \\
 & =\frac{1}{s_{1}s_{2}}(\Delta_{12}a_{2}\Delta_{21}a_{1}\sigma_{1}\sigma_{2}\sqrt{1-\rho^{2}}+\Delta_{12}a_{2}\Delta_{22}a_{2}\sigma_{2}^{2}\rho\sqrt{1-\rho^{2}})\nonumber \\
l_{11}^{s}l_{22}^{s}-l_{12}^{s}l_{21}^{s} & =\frac{1}{s_{1}s_{2}}(\Delta_{11}a_{1}\Delta_{22}a_{2}-\Delta_{12}a_{2}\Delta_{21}a_{1})\sigma_{1}\sigma_{2}\sqrt{1-\rho^{2}}\label{eq:lnumero}\\
l_{11}^{s}l_{21}^{s} & =\frac{1}{s_{1}s_{2}}(\Delta_{11}a_{1}\sigma_{1}+\Delta_{12}a_{2}\sigma_{2}\rho)(\Delta_{21}a_{1}\sigma_{1}+\Delta_{22}a_{2}\sigma_{2}\rho)\nonumber \\
 & =\frac{1}{s_{1}s_{2}}(\Delta_{11}a_{1}\Delta_{21}a_{1}\sigma_{1}^{2}+\Delta_{11}a_{1}\Delta_{22}a_{2}\sigma_{1}\sigma_{2}\rho+\Delta_{12}a_{2}\Delta_{21}a_{1}\sigma_{1}\sigma_{2}\rho+\Delta_{12}a_{2}\Delta_{22}a_{2}\sigma_{2}^{2}\rho^{2})\nonumber \\
l_{12}^{s}l_{22}^{s} & =\frac{1}{s_{1}s_{2}}\Delta_{12}a_{2}\Delta_{22}a_{2}\sigma_{2}^{2}(1-\rho^{2})\nonumber \\
l_{11}^{s}l_{21}^{s}+l_{12}^{s}l_{22}^{s} & =\frac{1}{s_{1}s_{2}}(\Delta_{11}a_{1}\Delta_{21}a_{1}\sigma_{1}^{2}+(\Delta_{11}a_{1}\Delta_{22}a_{2}+\Delta_{12}a_{2}\Delta_{21}a_{1})\sigma_{1}\sigma_{2}\rho+\Delta_{12}a_{2}\Delta_{22}a_{2}\sigma_{2}^{2})\label{eq:ldenom}
\end{align}
and therefore 
\begin{align}
\frac{1}{(\rho^{s})^{2}} & =1+\left(\frac{(\Delta_{11}a_{1}\Delta_{22}a_{2}-\Delta_{12}a_{2}\Delta_{21}a_{1})\sigma_{1}\sigma_{2}\sqrt{1-\rho^{2}}}{\Delta_{11}a_{1}\Delta_{21}a_{1}\sigma_{1}^{2}+(\Delta_{11}a_{1}\Delta_{22}a_{2}+\Delta_{12}a_{2}\Delta_{21}a_{1})\sigma_{1}\sigma_{2}\rho+\Delta_{12}a_{2}\Delta_{22}a_{2}\sigma_{2}^{2}}\right)^{2}\nonumber \\
 & =1+\left(\frac{(1-\frac{\Delta_{12}\Delta_{21}}{\Delta_{11}\Delta_{22}})\sqrt{1-\rho^{2}}}{\frac{\Delta_{21}a_{1}\sigma_{1}}{\Delta_{22}a_{2}\sigma_{2}}+(1+\frac{\Delta_{12}\Delta_{21}}{\Delta_{11}\Delta_{22}})\rho+\frac{\Delta_{12}a_{2}\sigma_{2}}{\Delta_{11}a_{1}\sigma_{1}}}\right)^{2}\label{eq:rho_S_Delta}
\end{align}
where we have used that $\Delta_{11},\Delta_{22},a_{1},a_{2},\sigma_{1},\sigma_{2}>0$.
As before, the sign of the correlation coefficient is given by the
sign of $l_{11}^{s}l_{21}^{s}+l_{12}^{s}l_{22}^{s}$.

Finally, the equity $\Delta$'s are given as (see appendix \ref{sec:Computing-the-Greeks}
for a detail computation): 
\begin{align}
  \label{eq:Delta11}
\Delta_{11} & =\Exq\left[\frac{1}{1-M_{12}^{s}M_{21}^{s}}\frac{A_{1,T}}{a_{1}}\Ind_{\Rss}+\frac{1}{1-M_{12}^{d}M_{21}^{s}}\frac{A_{1,T}}{a_{1}}\Ind_{\Rsd}\right]\\
& =\pss\frac{1}{1-M_{12}^{s}M_{21}^{s}}\Exq\left[\frac{A_{1,T}}{a_{1}}\bigm|\Rss\right]+\psd\frac{1}{1-M_{12}^{d}M_{21}^{s}}\Exq\left[\frac{A_{1,T}}{a_{1}}\bigm|\Rsd\right]\\
\label{eq:Delta12}
\Delta_{12} & =\Exq\left[\frac{1}{1-M_{12}^{s}M_{21}^{s}}M_{12}^{s}\frac{A_{2,T}}{a_{2}}\Ind_{\Rss}+\frac{1}{1-M_{12}^{d}M_{21}^{s}}M_{12}^{d}\frac{A_{2,T}}{a_{2}}\Ind_{\Rsd}\right]\\
& =\pss\frac{1}{1-M_{12}^{s}M_{21}^{s}}M_{12}^{s}\Exq\left[\frac{A_{2,T}}{a_{2}}\bigm|\Rss\right]+\psd\frac{1}{1-M_{12}^{d}M_{21}^{s}}M_{12}^{d}\Exq\left[\frac{A_{2,T}}{a_{2}}\bigm|\Rsd\right]\\
\label{eq:Delta21}
\Delta_{21} & =\Exq\left[\frac{1}{1-M_{12}^{s}M_{21}^{s}}M_{21}^{s}\frac{A_{1,T}}{a_{1}}\Ind_{\Rss}+\frac{1}{1-M_{12}^{s}M_{21}^{d}}M_{21}^{d}\frac{A_{1,T}}{a_{1}}\Ind_{\Rds}\right]\\
& =\pss\frac{1}{1-M_{12}^{s}M_{21}^{s}}M_{21}^{s}\Exq\left[\frac{A_{1,T}}{a_{1}}\bigm|\Rss\right]+\pds\frac{1}{1-M_{12}^{s}M_{21}^{d}}M_{21}^{d}\Exq\left[\frac{A_{1,T}}{a_{1}}\bigm|\Rds\right]\\
\label{eq:Delta22}
\Delta_{22} & =\Exq\left[\frac{1}{1-M_{12}^{s}M_{21}^{s}}\frac{A_{2,T}}{a_{2}}\Ind_{\Rss}+\frac{1}{1-M_{12}^{s}M_{21}^{d}}\frac{A_{2,T}}{a_{2}}\Ind_{\Rds}\right]\\
 & =\pss\frac{1}{1-M_{12}^{s}M_{21}^{s}}\Exq\left[\frac{A_{2,T}}{a_{2}}\bigm|\Rss\right]+\pds\frac{1}{1-M_{12}^{s}M_{21}^{d}}\Exq\left[\frac{A_{2,T}}{a_{2}}\bigm|\Rds\right]\;.
\end{align}

\subsection{Special cases and theorems}

\begin{theorem}\label{theo1} The equity correlation $\rho^{s}$ exceeds
the asset correlation $\rho$, i.e. $\rho^{s}\geq\rho$. \end{theorem}
\begin{proof} Here, we consider several cases. 
\begin{description}
\item [{$\rho = 0:$}] Then, by \eq{ldenom} 
\begin{align*}
l_{11}^{s}l_{21}^{s}+l_{12}^{s}l_{22}^{s} & =\frac{1}{s_{1}s_{2}}(\Delta_{11}a_{1}\Delta_{21}a_{1}\sigma_{1}^{2}+\Delta_{12}a_{2}\Delta_{22}a_{2}\sigma_{2}^{2})\geq0
\end{align*}
and therefore $\rho^{s}\geq0$ as well. 
\item [{$\rho > 0$:}] Then, by \eq{ldenom} we find that $l_{11}^{s}l_{21}^{s}+l_{12}^{s}l_{22}^{s}>0$
as well and therefore from \ref{eq:rho_s} $\rho^{s}>0$. Furthermore,
we compute 
\begin{alignat*}{2}
 &  & \rho^{s} & \geq\rho\\
\Leftrightarrow\quad &  & \frac{1}{(\rho^{s})^{2}} & \leq\frac{1}{\rho^{2}}\\
\Leftrightarrow\quad &  & 1+\left(\frac{x}{y}\right)^{2} & \leq\frac{1}{\rho^{2}}\\
\Leftrightarrow\quad &  & \left(\frac{x}{y}\right)^{2} & \leq\frac{1}{\rho^{2}}-1\\
\Leftrightarrow\quad &  & \left(\frac{x}{y}\right)^{2} & \leq\frac{1-\rho^{2}}{\rho^{2}}\\
\end{alignat*}
where $x=(1-\frac{\Delta_{12}\Delta_{21}}{\Delta_{11}\Delta_{22}})\sqrt{1-\rho^{2}}$
and $y=\frac{\Delta_{21}a_{1}\sigma_{1}}{\Delta_{22}a_{2}\sigma_{2}}+(1+\frac{\Delta_{12}\Delta_{21}}{\Delta_{11}\Delta_{22}})\rho+\frac{\Delta_{12}a_{2}\sigma_{2}}{\Delta_{11}a_{1}\sigma_{1}}$.
Continuing, we reason 
\begin{align*}
\frac{x^{2}}{y^{2}} & =\frac{(1-\frac{\Delta_{12}\Delta_{21}}{\Delta_{11}\Delta_{22}})^{2}(1-\rho^{2})}{\left(\frac{\Delta_{21}a_{1}\sigma_{1}}{\Delta_{22}a_{2}\sigma_{2}}+(1+\frac{\Delta_{12}\Delta_{21}}{\Delta_{11}\Delta_{22}})\rho+\frac{\Delta_{12}a_{2}\sigma_{2}}{\Delta_{11}a_{1}\sigma_{1}}\right)^{2}}\\
 & \leq\frac{(1-\frac{\Delta_{12}\Delta_{21}}{\Delta_{11}\Delta_{22}})^{2}}{(1+\frac{\Delta_{12}\Delta_{21}}{\Delta_{11}\Delta_{22}})^{2}}\cdot\frac{1-\rho^{2}}{\rho^{2}}\\
 & \leq\frac{1-\rho^{2}}{\rho^{2}}
\end{align*}
as required. 
\item [{$\rho < 0$:}] In this case, whenever $l_{11}^{s}l_{21}^{s}+l_{12}^{s}l_{22}^{s}\geq0$
we clearly have $\rho^{s}\geq0>\rho$. Thus, we assume that 
\begin{alignat*}{2}
 &  & l_{11}^{s}l_{21}^{s}+l_{12}^{s}l_{22}^{s} & <0\\
\Leftrightarrow\quad &  & \Delta_{11}a_{1}\Delta_{21}a_{1}\sigma_{1}^{2}+(\Delta_{11}a_{1}\Delta_{22}a_{2}+\Delta_{12}a_{2}\Delta_{21}a_{1})\sigma_{1}\sigma_{2}\rho+\Delta_{12}a_{2}\Delta_{22}a_{2}\sigma_{2}^{2} & <0\\
\Leftrightarrow\quad &  & \frac{\Delta_{11}a_{1}\Delta_{21}a_{1}\sigma_{1}^{2}+\Delta_{12}a_{2}\Delta_{22}a_{2}\sigma_{2}^{2}}{(\Delta_{11}a_{1}\Delta_{22}a_{2}+\Delta_{12}a_{2}\Delta_{21}a_{1})\sigma_{1}\sigma_{2}} & <-\rho\\
\Leftrightarrow\quad &  & \frac{\frac{\Delta_{21}a_{1}\sigma_{1}}{\Delta_{22}a_{2}\sigma_{2}}+\frac{\Delta_{12}a_{2}\sigma_{2}}{\Delta_{11}a_{1}\sigma_{1}}}{1+\frac{\Delta_{12}\Delta_{21}}{\Delta_{11}\Delta_{22}}} & <-\rho\;.
\end{alignat*}
Now, defining $x=\frac{\Delta_{12}\Delta_{21}}{\Delta_{11}\Delta_{22}}$
and $y=\frac{\Delta_{21}a_{1}\sigma_{1}}{\Delta_{22}a_{2}\sigma_{2}}+\frac{\Delta_{12}a_{2}\sigma_{2}}{\Delta_{11}a_{1}\sigma_{1}}$
the above equation reads 
\begin{align}
\frac{y}{1+x}<-\rho\quad\Leftrightarrow\quad y & <-\rho(1+x)\;.\label{eq:negrho}
\end{align}
The desired result $\rho^{s}\geq\rho$ than follows if we can show
that 
\begin{alignat*}{2}
 & \frac{1}{(\rho^{s})^{2}} & =1+\frac{(1-x)^{2}(1-\rho^{2})}{(y+(1+x)\rho)^{2}} & \geq\frac{1}{\rho^{2}}\\
\Leftrightarrow\quad &  & \frac{(1-x)^{2}(1-\rho^{2})}{\left(y+(1+x)\rho\right)^{2}} & \geq\frac{1-\rho^{2}}{\rho^{2}}\\
\Leftrightarrow\quad &  & \left(y+(1+x)\rho\right)^{2} & \leq(1-x)^{2}\rho^{2}\\
\Leftrightarrow\quad &  & y^{2}+2y(1+x)\rho+\rho^{2}\left(1-(1-x)^{2}\right) & \leq0\:.
\end{alignat*}
From \eq{negrho} we have that 
\begin{align*}
y^{2}+2y(1+x)\rho+\rho^{2}\left(1-(1-x)^{2}\right) & <\rho^{2}(1+x)^{2}-2(1+x)\rho(1+x)\rho+\rho^{2}\left(1-(1-x)^{2}\right)\\
 & =\rho^{2}\left(1-(1-x)^{2}-(1+x)^{2}\right)\\
 & =-\rho^{2}(1+2x^{2})
\end{align*}
which is obviously negative and thereby completes the proof. 
\end{description}
\end{proof}
\begin{description}
\item [{Higher}] values: Consider the limit $a_{1},a_{2}\to\infty$. Then, 
\end{description}
\begin{align}
\frac{\Delta_{12}}{\Delta_{11}} & =\frac{\pss\frac{1}{1-M_{12}^{s}M_{21}^{s}}M_{12}^{s}\Exq\left[\frac{A_{2,T}}{a_{2}}\bigm|\Rss\right]+\psd\frac{1}{1-M_{12}^{d}M_{21}^{s}}M_{12}^{d}\Exq\left[\frac{A_{2,T}}{a_{2}}\bigm|\Rsd\right]}{\pss\frac{1}{1-M_{12}^{s}M_{21}^{s}}\Exq\left[\frac{A_{1,T}}{a_{1}}\bigm|\Rss\right]+\psd\frac{1}{1-M_{12}^{d}M_{21}^{s}}\Exq\left[\frac{A_{1,T}}{a_{1}}\bigm|\Rsd\right]}\nonumber \\
 & =\frac{\frac{1}{1-M_{12}^{s}M_{21}^{s}}M_{12}^{s}\Exq\left[\frac{A_{2,T}}{a_{2}}\bigm|\Rss\right]+\frac{\psd}{\pss}\frac{1}{1-M_{12}^{d}M_{21}^{s}}M_{12}^{d}\Exq\left[\frac{A_{2,T}}{a_{2}}\bigm|\Rsd\right]}{\frac{1}{1-M_{12}^{s}M_{21}^{s}}\Exq\left[\frac{A_{1,T}}{a_{1}}\bigm|\Rss\right]+\frac{\psd}{\pss}\frac{1}{1-M_{12}^{d}M_{21}^{s}}\Exq\left[\frac{A_{1,T}}{a_{1}}\bigm|\Rsd\right]}\nonumber \\
 & \xrightarrow[a_{1},a_{2}\to\infty]{}M_{12}^{s}\frac{\Exq\left[\frac{A_{2,T}}{a_{2}}\bigm|\Rss\right]}{\Exq\left[\frac{A_{1,T}}{a_{1}}\bigm|\Rss\right]}\label{eq:DeltaRatio1}
\end{align}
as $\psd\to0,\pss\to1$ for $a_{1},a_{2}\to\infty$ and the conditional
expectations $\Exq\left[\frac{A_{1,T}}{a_{1}}\bigm|\Rss\right],\Exq\left[\frac{A_{2,T}}{a_{2}}\bigm|\Rss\right]$
are constants which do not depend on $a_{1},a_{2}$.

Similarly, 
\begin{align}
\frac{\Delta_{21}}{\Delta_{22}} & =\frac{\pss\frac{1}{1-M_{12}^{s}M_{21}^{s}}M_{21}^{s}\Exq\left[\frac{A_{1,T}}{a_{1}}\bigm|\Rss\right]+\pds\frac{1}{1-M_{12}^{s}M_{21}^{d}}M_{21}^{d}\Exq\left[\frac{A_{1,T}}{a_{1}}\bigm|\Rds\right]}{\pss\frac{1}{1-M_{12}^{s}M_{21}^{s}}\Exq\left[\frac{A_{2,T}}{a_{2}}\bigm|\Rss\right]+\pds\frac{1}{1-M_{12}^{s}M_{21}^{d}}\Exq\left[\frac{A_{2,T}}{a_{2}}\bigm|\Rds\right]}\nonumber \\
 & \xrightarrow[a_{1},a_{2}\to\infty]{}M_{21}^{s}\frac{\Exq\left[\frac{A_{1,T}}{a_{1}}\bigm|\Rss\right]}{\Exq\left[\frac{A_{2,T}}{a_{2}}\bigm|\Rss\right]}\;.\label{eq:DeltaRatio2}
\end{align}
Yet, in general, the limit of $\rho^{s}$ does not exist as it depends
explicitly on $\frac{a_{1}}{a_{2}}$ and thereby on the specific path
along which $a_{1},a_{2}$ are taken to infinity.

\subsubsection{Debt cross-holdings only}

Assuming debt cross-holdings only, i.e. $\v{M}^{s}\equiv\v{0}$, we
obtain the following proposition.

\begin{proposition} Assuming debt cross-holdings only, we have 
\[
\rho^{s}\xrightarrow[a_{1},a_{2}\to\infty]{}\rho\;.
\]
\end{proposition} \begin{proof} First we show that in the limit
$a_{1},a_{2}\to\infty$ we have $(\rho^{s})^{2}\xrightarrow[a_{1},a_{2}\to\infty]{}\rho^{2}$.
For this note that from the considerations above (\eq{DeltaRatio1}
and \eq{DeltaRatio2}), 
\begin{align*}
\frac{\Delta_{12}}{\Delta_{11}} & \xrightarrow[a_{1},a_{2}\to\infty]{}0\\
\frac{\Delta_{21}}{\Delta_{22}} & \xrightarrow[a_{1},a_{2}\to\infty]{}0\\
\end{align*}
and thus 
\begin{align*}
\frac{1}{(\rho^{s})^{2}} & =1+\left(\frac{(1-\frac{\Delta_{12}\Delta_{21}}{\Delta_{11}\Delta_{22}})\sqrt{1-\rho^{2}}}{\frac{\Delta_{21}a_{1}\sigma_{1}}{\Delta_{22}a_{2}\sigma_{2}}+(1+\frac{\Delta_{12}\Delta_{21}}{\Delta_{11}\Delta_{22}})\rho+\frac{\Delta_{12}a_{2}\sigma_{2}}{\Delta_{11}a_{1}\sigma_{1}}}\right)^{2}\\
 & \xrightarrow[a_{1},a_{2}\to\infty]{}1+\left(\frac{\sqrt{1-\rho^{2}}}{0+\rho+0}\right)^{2}\\
 & =\frac{\rho^{2}+1-\rho^{2}}{\rho^{2}}=\frac{1}{\rho^{2}}
\end{align*}
along any path at which $\frac{a_{1}}{a_{2}}$ and $\frac{a_{2}}{a_{1}}$
stay bounded. 

Furthermore, by the same argument $\frac{\Delta_{21}a_{1}\sigma_{1}}{\Delta_{22}a_{2}\sigma_{2}}+(1+\frac{\Delta_{12}\Delta_{21}}{\Delta_{11}\Delta_{22}})\rho+\frac{\Delta_{12}a_{2}\sigma_{2}}{\Delta_{11}a_{1}\sigma_{1}}\xrightarrow[a_{1},a_{2}\to\infty]{}\rho$
meaning that the signs of $\rho^{s}$ and $\rho$ eventually agree
and we conclude that $\rho^{s}\to\rho$. \end{proof}

\subsubsection{Merton model}

Sanity check of special cases 
\begin{itemize}
\item \label{No-network,}No network; i.e., $\v{M}^{s}=\v{M}^{d}\equiv\v{0}$.
Then, 
\begin{align*}
\frac{l_{22}^{s}}{l_{21}^{s}} & =\frac{\left(\pss\Exq[A_{2,T}|\Rss]+\pds\Exq[A_{2,T}|\Rds]\right)\sigma_{2}\sqrt{1-\rho^{2}}}{\left(\pss\Exq[A_{2,T}|\Rss]+\pds\Exq[A_{2,T}|\Rds]\right)\sigma_{2}\rho}=\frac{\sqrt{1-\rho^{2}}}{\rho}\\
\rho^{s} & =\frac{1}{\sqrt{1+\left(\frac{1-\rho^{2}}{\rho^{2}}\right)}}=\rho
\end{align*}
\item Volatility and the Merton model. According to \eq{cholesky} we have
$\sigma_{1}=l_{11}$ and therefore using \eq{Lexplizit} again 
\begin{align*}
\sigma_{1}^{s} & =l_{11}^{s}=\frac{\Delta_{11}a_{1}}{s_{1}}l_{11}+\frac{\Delta_{12}a_{2}}{s_{1}}l_{21}\\
 & =\frac{\left(\pss\frac{\Exq[A_{1,T}|\Rss]}{1-M_{12}^{s}M_{21}^{s}}+\psd\frac{\Exq[A_{1,T}|\Rsd]}{1-M_{12}^{d}M_{21}^{s}}\right)\sigma_{1}+\left(\pss\frac{M_{12}^{s}\Exq[A_{2,T}|\Rss]}{1-M_{12}^{s}M_{21}^{s}}+\psd\frac{M_{12}^{d}\Exq[A_{2,T}|\Rsd]}{1-M_{12}^{d}M_{21}^{s}}\right)\sigma_{2}\rho}{\pss\frac{\Exq[A_{1,T}|\Rss]-d_{1}+M_{12}^{d}d_{2}+M_{12}^{s}(\Exq[A_{2,T}|\Rss]-d_{2}+M_{21}^{d}d_{1})}{1-M_{12}^{s}M_{21}^{s}}+\psd\frac{\Exq[A_{1,T}|\Rsd]-d_{1}+M_{12}^{d}d_{2}+M_{12}^{d}(\Exq[A_{2,T}|\Rsd]-d_{2}+M_{21}^{d}d_{1})}{1-M_{12}^{d}M_{21}^{s}}}
\end{align*}
which should reduce to the standard Merton model formulas without
a network: 
\begin{align*}
\sigma_{1}^{s} & =\frac{\left(\pss\Exq[A_{1,T}|\Rss]+\psd\Exq[A_{1,T}|\Rsd]\right)\sigma_{1}}{\pss(\Exq[A_{1,T}|\Rss]-d_{1})+\psd(\Exq[A_{1,T}|\Rsd]-d_{1})}\\
 & =\frac{\pi_{s\cdot}\Exq[A_{1,T}|A_{s\cdot}]}{\pi_{s\cdot}\Exq[A_{1,T}|A_{s\cdot}]-\pi_{s\cdot}d_{1}}\sigma_{1}=\frac{\sigma_{1}}{1-\frac{d_{1}}{\Exq[A_{1,T}|A_{s\cdot}]}}\;.
\end{align*}
Translating to more standard notation, i.e. $A_{T}=A_{1,T}$ and $K=d_{1}$,
and using that the conditional expectation $\Exq[A_{T}|A_{T}\geq K]$
can be computed as 
\begin{align*}
\Exq[A_{T}|A_{T}\geq K] & =e^{\mu_{BS}+\frac{\sigma_{BS}^{2}}{2}}\frac{\Phi(\frac{\mu_{BS}+\sigma_{BS}^{2}-\ln K}{\sigma_{BS}})}{\Phi(\frac{\mu_{BS}-\ln K}{\sigma_{BS}})}\\
 & =a_{t}e^{r(T-t)}\frac{\Phi(d_{+})}{\Phi(d_{-})}
\end{align*}
\end{itemize}
where $\mu_{BS}=(r-\frac{\sigma^{2}}{2})(T-t)+\ln a_{t}$ and $\sigma_{BS}=\sigma\sqrt{T-t}$.
Further, $d_{\pm}$ denotes the familiar terms 
\[
d_{\pm}=\frac{\ln\frac{a_{t}}{K}+(r\pm\frac{\sigma^{2}}{2})(T-t)}{\sigma\sqrt{T-t}}
\]
and $\Phi$ the cumulative distribution function of a standard normal.
Thus, plugging everything together we obtain the standard result 
\begin{align*}
\sigma^{s} & =\frac{\sigma}{1-\frac{d_{1}}{\Exq[A_{1,T}|A_{s\cdot}]}}\\
 & =\frac{\sigma}{1-\frac{K\Phi(d_{-})}{a_{t}e^{r(T-t)}\Phi(d_{+})}}\\
 & =\frac{a_{t}\Phi(d_{+})}{a_{t}\Phi(d_{+})-e^{-r(T-t)}K\Phi(d_{-})}\sigma\\
 & =\Delta_{BS}\frac{a_{t}}{c_{BS}}\sigma
\end{align*}
with the Black-Scholes Delta $\Delta_{BS}=\Phi(d_{+})$ and call price
$c_{BS}=a_{t}\Phi(d_{+})-e^{-r(T-t)}K\Phi(d_{-})$.

\subsection{Numerical illustrations}

\subsubsection{Equity correlation as function of the network parameters}

One of the central results of this paper is given by the Theorem
\ref{theo1}; i.e., $\rho^{s}\geq\rho$ for two firms. Here, we
illustrate it numerically by plotting the equity correlation for
different values of asset correlations, initial prices, volatilities
and cross-holding fractions. For the sake of simplicity, we consider
cross-holdings of debt only ($\v{M}^{s}=\v{0}$) and symmetric initial
conditions for the assets\footnote{Thus, the both firms assets have
the same spot value, but are still log-normally distributed at
maturity. The case of comonotonic asset endowments as considered by
\cite{feinstein2021} corresponds to the trivial case of fully
correlated assets, i.e., $\rho = 1$, and thus $\rho_s = 1$ as well.},
i.e., $\sigma_{1}=\sigma_{2}=\sigma$ and
$a_{\text{1,0}}=a_{2,0}=a_{1,2}$.  \Fig{rhoS_1} show the resulting
equity correlation as a function of firm 1's equity\footnote{ Note
that by symmetry of the setup, the corresponding figure for firm 2's
equity is just the mirror image, i.e., obtained by exchanging
$M_{12}^{d}$ and $M_{21}^{d}$.}, for different debt cross-holding
fractions, asset correlations and volatilities.  Here, the subplots
correspond to cross-holding fractions of $M_{12}^{d} = 0, 0.2, \ldots,
0.8$ (vertical) and $M_{21}^{d} = 0, 0.2, \ldots, 0.8$
(horizontal). Furthermore, solid lines represent an asset volatility
value of $0.2$ while dotted lines a volatility value of 0.4. The
colors correspond to different values of asset correlations
$\rho_{a}=\left\{ -0.4,0,0.4,0.8\right\} $.  As proved above (page
\pageref{No-network,}), in case of no cross-holdings (left and upper
subplot) we find $\rho^{s}=\rho$. In contrast, with cross-holdings the
equity correlation shows a marked increase above the asset
correlations until the equity reaches essentially zero. Most notable,
even for anti-correlated business asset ($\rho_{a}=-0.4$) the firm's
equities exhibit positive correlations for sufficiently large
cross-holding fractions and stressed firm equities, i.e., during
crises times with correspondingly low asset values.  Similar effects
are also observed for asymmetric asset values and with additional
equity cross-holdings (as shown in appendix \ref{app:morefigs}).

\begin{figure}[ht]
  \centering \includegraphics[width=0.9\textwidth]{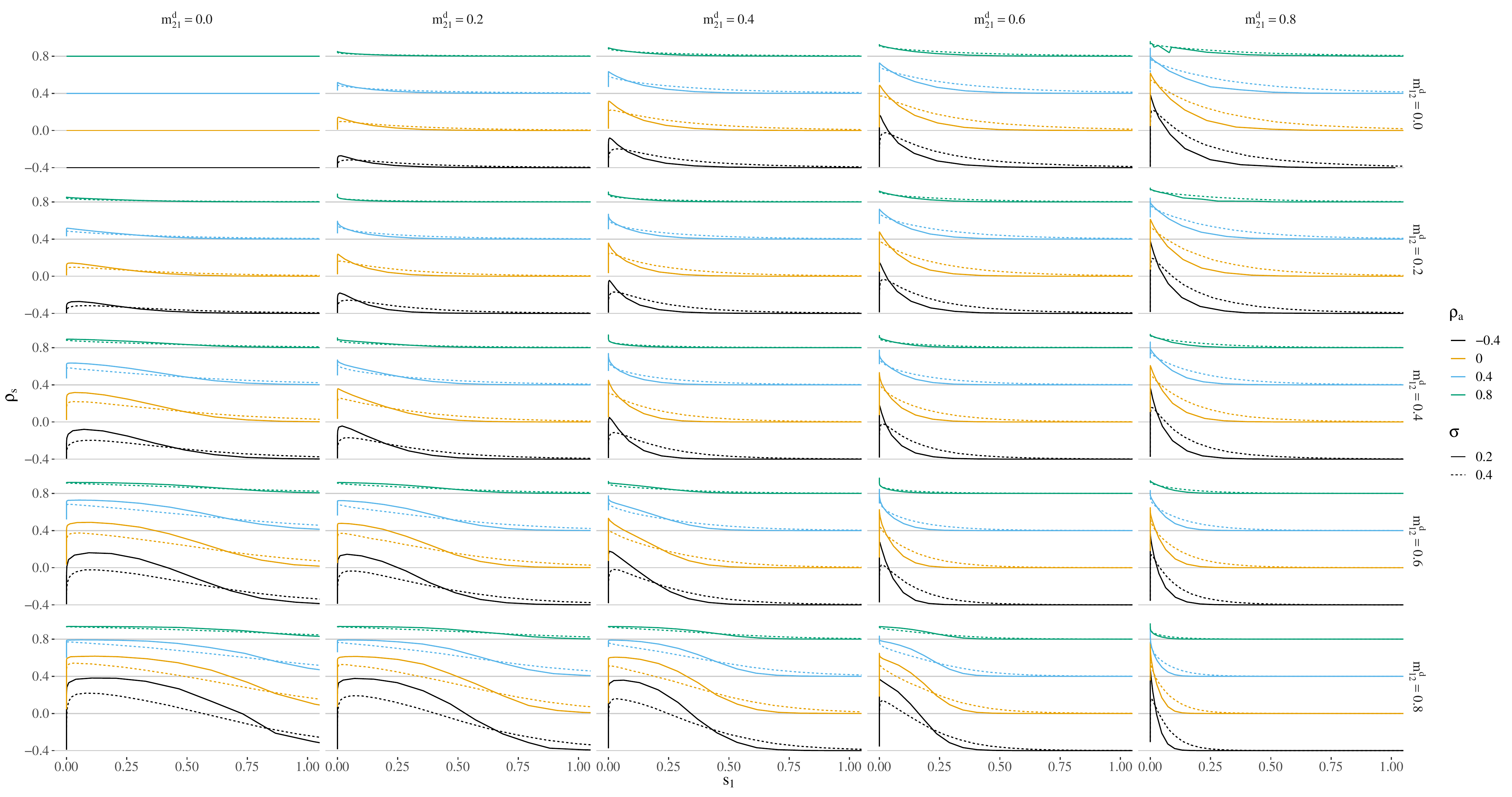} \caption{Equity correlations as function of firm 1's equity value $s_{1,0}$,
    for different debt cross-holding fractions, asset correlations and volatilities.
\label{fig:rhoS_1} }
\end{figure}

\section{Summary}

We have used a financial network with cross-holdings to model the
complex interlinkages around the financial firms and we put the accent
in the study of the correlation of their derivatives. In fact, we
uncover the capabilities of the \citeauthor{suzuki2002valuing} model
to address the non-constant behavior of the correlation observed in
the equity market; even from constant values of the business asset
correlations. We shown mathematically, that the correlation among
equities depends on the structure of the financial networks. Furthermore,
we demonstrate analytically for the two firms case, the equity correlation
is never lower than the unconditional correlation of the asset returns.
Besides, the numerical simulations shows the power of the network
approach too explain structurally the increase in correlation under
crisis.

\section*{Acknowledgement}

This work has been supported by DFG BE 7225/1--1.

\bibliographystyle{apacite}
\bibliography{Correlations}

\appendix

\section{Computing the Greeks\label{sec:Computing-the-Greeks}}

The Greeks quantify the sensitivities of derivative prices to changes
in underlying parameters. Here, we consider first-order Greeks only.
In particular, we compute the sensitivities of equity and debt prices
accounting for cross-holdings with respect to current asset values
$\Delta=\frac{\partial\v{x}_{t}}{\partial\v{a}_{t}}$.

First, we recall the solution $\v{A}_{t}$ of \eq{SDE} as 
\begin{align}
\v{A}_{t} & =\v{a}_{0}e^{\left(r-\frac{1}{2}\diag(\v{\sigma}^{2})\right)t+\diag(\v{\sigma})\v{W}_{t}}\label{eq:p_A_t-1}
\end{align}
where $\v{a}_{0}>0$ denotes the initial value and $\v{W}_{t}$ is
multivariate normal distributed with mean $\v{0}$ and covariance
matrix $t\v{C}$. Note that $\v{W}_{t}$ can be obtained from independent
standard normal variates $\v{Z}\sim\mathcal{N}(\v{0},\v{I}_{n\times n})$
as $\v{W}_{t}=\sqrt{t}\v{L}\v{Z}$ with $\v{L}^{T}\v{L}=\v{C}$. We
will use this representation in the next section to express the risk-neutral
market value of equity and debt contracts as 
\begin{align}
\v{x}_{t}=\E_{t}^{Q}[e^{-r\tau}\v{x}^{*}\left(\v{a}_{T}(\v{Z})\right)]=\E_{t}^{Q}\left[e^{-r\tau}\v{x}^{*}\left(\v{a}_{t}e^{\left(r-\frac{1}{2}\v{\sigma}^{2}\right)\tau+\sqrt{\tau}\diag(\v{\sigma})\v{L}\v{Z}}\right)\right]\,.\label{eq:val_Q-1}
\end{align}

\subsection{Formal solution}

Denoting all parameters of interest by $\v{\theta}=(\v{a}_{t},\v{\sigma},r,\tau)^{T}$
and considering that the asset value $\v{a}_{\tau}(Z;\theta)$ depends
on the random variate $Z$ and these parameters, we need to compute
the following derivatives 
\begin{align}
\frac{\partial}{\partial\v{\theta}}\v{x}_{t} & =\frac{\partial}{\partial\v{\theta}}\E_{t}^{Q}[e^{-r\tau}\v{x}^{*}(\v{a}_{T}(Z;\theta))] \nonumber \\
 & =\E_{t}^{Q}\left[\left(\frac{\partial}{\partial\v{\theta}}e^{-r\tau}\right)\v{x}^{*}(\v{a}_{T}(Z;\theta))+e^{-r\tau}\left(\frac{\partial}{\partial\v{\theta}}\v{x}^{*}(\v{a}_{T}(Z;\theta))\right)\right] \label{eq:greek_Q}
\end{align}
where we have used pathwise differentiation. Exchanging integration
and differentiation requires some continuity conditions on $\v{x}^{*}$.
In particular, \citeA[proposition 1]{broadie1996estimating} prove
that pathwise differentiation is applicable for Lipschitz continuous
functions.

\begin{lemma} The function $\v{x}^{*}(\v{a})$ is Lipschitz continuous
with Lipschitz constant 
\begin{align*}
L^{*}=(1-\max_{\v{\xi}}\norm{\v{K}_{\v{\xi}}})^{-1}
\end{align*}
where 
\begin{align*}
\v{K}_{\v{\xi}} & =\left(\begin{array}{cc}
\diag(\v{\xi})\v{M}^{s} & \diag(\v{\xi})\v{M}^{d}\\
\diag(\v{1}-\v{\xi})\v{M}^{s} & \diag(\v{1}-\v{\xi})\v{M}^{d}
\end{array}\right)\,.
\end{align*}
\end{lemma}

\begin{proof} First, we observe that $\v{g}$ is continuous and piecewise
linear. In particular, we have 
\begin{align*}
\norm{\v{g}(\v{a}_{1},\v{x})-\v{g}(\v{a}_{2},\v{x})} & \leq\norm{\v{a}_{1}-\v{a}_{2}}\\
\norm{\v{g}(\v{a},\v{x}_{1})-\v{g}(\v{a},\v{x}_{2})} & \leq\max_{\v{\xi}}\norm{\v{K}_{\v{\xi}}}\norm{\v{x}_{1}-\v{x}_{2}}
\end{align*}
showing that $\v{g}$ is Lipschitz continuous with respect to $\v{a}$
and $\v{x}$.

Note that by assumption $1$, we have $\norm{\v{M}^{s}},\norm{\v{M}^{d}}<1$.
Furthermore, the solvency indicator is either $\xi_{i}=1$ if bank
$i$ is solvent or $\xi_{i}=0$ otherwise. Thus, it holds that $\norm{\v{K}_{\v{\xi}}}<1,\forall\v{\xi}\in\{0,1\}^{N}$.

Then, we compute 
\begin{align*}
\norm{\v{x}^{*}(\v{a}_{1})-\v{x}^{*}(\v{a}_{2})} & =\norm{\v{g}(\v{a}_{1},\v{x}^{*}(\v{a}_{1}))-\v{g}(\v{a}_{2},\v{x}^{*}(\v{a}_{2}))}\\
 & =\norm{\v{g}(\v{a}_{1},\v{x}^{*}(\v{a}_{1}))-\v{g}(\v{a}_{1},\v{x}^{*}(\v{a}_{2}))+\v{g}(\v{a}_{1},\v{x}^{*}(\v{a}_{2}))-\v{g}(\v{a}_{2},\v{x}^{*}(\v{a}_{2}))}\\
 & \leq\norm{\v{g}(\v{a}_{1},\v{x}^{*}(\v{a}_{1}))-\v{g}(\v{a}_{1},\v{x}^{*}(\v{a}_{2}))}+\norm{\v{g}(\v{a}_{1},\v{x}^{*}(\v{a}_{2}))-\v{g}(\v{a}_{2},\v{x}^{*}(\v{a}_{2}))}\\
 & \leq\max_{\v{\xi}}\norm{\v{K}_{\v{\xi}}}\norm{\v{x}^{*}(\v{a}_{1}))-\v{x}^{*}(\v{a}_{2}))}+\norm{\v{a}_{1}-\v{a}_{2}}\\
\implies\norm{\v{x}^{*}(\v{a}_{1})-\v{x}^{*}(\v{a}_{2})} & \leq(1-\max_{\v{\xi}}\norm{\v{K}_{\v{\xi}}})^{-1}\norm{\v{a}_{1}-\v{a}_{2}}\;.
\end{align*}
\end{proof}

By the chain rule of differentiation we obtain 
\begin{align}
\frac{\partial}{\partial\v{\theta}}\v{x}^{*}(\v{a}_{\tau}(Z;\theta)) & =\frac{\partial}{\partial\v{a}}\v{x}^{*}(\v{a})\mid_{\v{a}=\v{a}_{\tau}(Z;\theta)}\frac{\partial}{\partial\v{\theta}}\v{a}_{\tau}(Z;\theta)\,.\label{eq:greek_chain}
\end{align}
Note that $\frac{\partial}{\partial\v{a}}\v{x}^{*}(\v{a})$ is the
derivative of the fixed point solving (\ref{eq:XOS_fix}). In order
to compute it, we make use of the {\em implicit function theorem}.
A version of the theorem by \citeA{halkin1974implicit} is adopted
to our notation: \begin{theorem} \label{thm:impl_fun} Let $U\subset\mathbb{R}^{m},V\subset\mathbb{R}^{n}$
and $\v{f}:U\times V\to\mathbb{R}^{n}$ a continuously differentiable
function. Suppose that 
\begin{align}
\v{f}(\v{x}^{*},\v{y}^{*})=\v{0}\label{eq:impl_fix}
\end{align}
at a point $(\v{x}^{*},\v{y}^{*})\in U\times V$ and that the Jacobian
matrices $\v{J}_{\v{f},\v{x}}\v{f}(\v{x},\v{y}),\v{J}_{\v{f},\v{y}}\v{f}(\v{x},\v{y})$
of partial derivatives exist at $(\v{x}^{*},\v{y}^{*})$. Further,
$\v{J}_{\v{f},\v{y}}$ is invertible at this point. Then, there exists
a neighborhood $U^{*}\subset U$ and a continuously differentiable
function $\v{h}:U^{*}\to\mathbb{R}^{n}$ with 
\begin{align}
\v{h}(\v{x}^{*}) & =\v{y}^{*}\label{eq:impl_sol}
\end{align}
and 
\begin{align}
\v{f}(\v{x},\v{h}(\v{x}))=\v{0}\quad\forall\v{x}\in U^{*}\,.\label{eq:impl_sol_fix}
\end{align}
Moreover, the partial derivatives of $\v{h}$ with respect to $\v{x}\in U^{*}$
are given as 
\begin{align}
\frac{\partial}{\partial\v{x}}\v{h}(\v{x}) & =-\left[\v{J}_{\v{f},\v{y}}\v{f}(\v{x},\v{h}(\v{x}))\right]_{n\times n}^{-1}\left[\frac{\partial}{\partial\v{x}}\v{f}(\v{x},\v{h}(\v{x}))\right]_{n\times m}\label{eq:impl_sol_deriv}
\end{align}
\end{theorem} As the function $\v{g}(\v{a},\v{x})$ defined in \eq{XOS_g}
is Lipschitz continuous, it is almost everywhere differentiable. The
partial derivatives are given by 
\begin{align}
\frac{\partial}{\partial s_{j}}g_{i}^{s}(\v{a},\v{x}) & =\left\{ \begin{array}{rp{1cm}lcccc}
M_{ij}^{s} &  & \mbox{if firm \ensuremath{i} is solvent}\\
0 &  & \mbox{otherwise}
\end{array}\right.\label{eq:g_diff}\\
\frac{\partial}{\partial s_{j}}g_{i}^{r}(\v{a},\v{x}) & =\left\{ \begin{array}{rp{1cm}lcccc}
0 &  & \mbox{if firm \ensuremath{i} is solvent}\\
M_{ij}^{s} &  & \mbox{otherwise}
\end{array}\right.\\
\frac{\partial}{\partial r_{j}}g_{i}^{s}(\v{a},\v{x}) & =\left\{ \begin{array}{rp{1cm}lcccc}
M_{ij}^{d} &  & \mbox{if firm \ensuremath{i} is solvent}\\
0 &  & \mbox{otherwise}
\end{array}\right.\\
\frac{\partial}{\partial r_{j}}g_{i}^{r}(\v{a},\v{x}) & =\left\{ \begin{array}{rp{1cm}lcccc}
0 &  & \mbox{if firm \ensuremath{i} is solvent}\\
M_{ij}^{d} &  & \mbox{otherwise}
\end{array}\right.\\
\frac{\partial}{\partial a_{j}}g_{i}^{s}(\v{a},\v{x}) & =\left\{ \begin{array}{rp{1cm}lcccc}
\phantom{M}1 &  & \mbox{if \ensuremath{i=j} and firm \ensuremath{i} is solvent}\\
0 &  & \mbox{otherwise}
\end{array}\right.\\
\frac{\partial}{\partial a_{j}}g_{i}^{r}(\v{a},\v{x}) & =\left\{ \begin{array}{rp{1cm}lcccc}
\phantom{M}0 &  & \mbox{if \ensuremath{i=j} and firm \ensuremath{i} is solvent}\\
1 &  & \mbox{otherwise}
\end{array}\right.\,.
\end{align}
Here, a firm $i$ is solvent if its asset value $v_{i}$ is sufficient
to repay its nominal debt $d_{i}$, i.e. $v_{i}=a_{i}+\sum_{j=1}^{n}M_{ij}^{s}s_{j}+\sum_{j=1}^{n}M_{ij}^{d}r_{j}>d_{i}$.
The derivatives of $\v{g}$ exist everywhere except for the boundary
case $v_{i}=d_{i}$. Defining the solvency vector $\v{\xi}=(\mathbbm{1}_{v_{i}>d_{1}}(v_{1}),\ldots,\mathbbm{1}_{v_{n}>d_{n}}(v_{n}))$,
the partial derivatives of $\v{g}$ with respect to $\v{x}$ can be
collected in a matrix as follows 
\begin{align}
\frac{\partial}{\partial\v{x}}\v{g}(\v{a},\v{x}) & =\left[\begin{array}{ccc}
\diag(\v{\xi})\v{M}^{s} &  & \diag(\v{\xi})\v{M}^{d}\\
\\
\diag(\v{1}_{n}-\v{\xi})\v{M}^{s} &  & \diag(\v{1}_{n}-\v{\xi})\v{M}^{d}
\end{array}\right]\label{eq:g_diff_mat}\\
 & =\diag\left((\v{\xi};\v{1}_{n}-\v{\xi})\right)\left[\begin{array}{ccc}
\v{M}^{s} &  & \v{M}^{d}\\
\\
\v{M}^{s} &  & \v{M}^{d}
\end{array}\right]
\end{align}

Thus, defining $\v{f}(\v{a},\v{x})=\v{x}-\v{g}(\v{a},\v{x})$ we obtain
by the implicit function theorem \ref{thm:impl_fun} \begin{corollary}
\label{coro:fix_diff} The partial derivatives of $\v{x}^{*}(\v{a})$
are given by 
\begin{align}
\frac{\partial}{\partial\v{a}}\v{x}^{*}(\v{a}) & =\left[\v{I}_{2n\times2n}-\frac{\partial}{\partial\v{x}}\v{g}(\v{a},\v{x})\right]^{-1}\left[\begin{array}{c}
\diag(\v{\xi})\\
\\
\diag(\v{1}_{n}-\v{\xi})
\end{array}\right]\label{eq:fix_diff}
\end{align}
\end{corollary} \begin{proof} Use that $\frac{\partial}{\partial\v{x}}\v{f}(\v{a},\v{x})=\v{I}_{2n\times2n}-\frac{\partial}{\partial\v{x}}\v{g}(\v{a},\v{x})$
and $\frac{\partial}{\partial\v{a}}\v{f}(\v{a},\v{x})=-\frac{\partial}{\partial\v{a}}\v{g}(\v{a},\v{x})$.
Then, the result follows from theorem \ref{thm:impl_fun} and $\frac{\partial}{\partial\v{a}}\v{g}(\v{a},\v{x})=\left[\begin{array}{c}
\diag(\v{\xi})\\
\diag(\v{1}_{n}-\v{\xi})
\end{array}\right]$. As explained below, assumption \ref{assu_1} ensures that $\frac{\partial}{\partial\v{x}}\v{f}(\v{a},\v{x})$
is invertible as required. \end{proof}

Finally, combining \eq{greek_Q} and (\ref{eq:greek_chain}) with
corollary \ref{coro:fix_diff} we formally compute the network Greeks
as 
\begin{align}
\frac{\partial}{\partial\v{\theta}}\v{x}_{t} & =\E_{t}^{Q}\left[\left(\frac{\partial}{\partial\v{\theta}}e^{-r\tau}\right)\v{x}^{*}(\v{a}_{T}(Z;\theta))\right.\nonumber \\
 & \phantom{=}\left.+e^{-r\tau}\left[\v{I}_{2n\times2n}-\frac{\partial}{\partial\v{x}}\v{g}(\v{a},\v{x})\right]^{-1}\left[\begin{array}{c}
\diag(\v{\xi})\\
\\
\diag(\v{1}_{n}-\v{\xi})
    \end{array}\right]\frac{\partial}{\partial\v{\theta}}\v{a}_{T}(Z;\theta)\right]
\label{eq:net_greeks_formal}
\end{align}
where the expectation is well-defined as the derivatives exist almost
everywhere, i.e. except for a set of measure zero.

\subsection{Two bank Delta}

In case of two banks, the network Delta, i.e.~$\v{\Delta} =
\frac{\partial \v{x}_t}{\partial \v{a}_t}$, can be computed
explicitly. First, we drop the time index $t$ to ease notation and
denote firm values (of equity and debt) and asset prices at time $t$
as $\v{x} = (s_1, s_2, r_1, r_2)$ and $\v{a} = (a_1, a_2)$
respectively. Then, using that
$\frac{\partial}{\partial\v{a}}e^{-r\tau} = \v{0}$ and
\begin{align*}
  \frac{\partial}{\partial\v{a}}\v{A}_{T} &=
  \frac{\partial}{\partial\v{a}}\v{a} \; e^{\left(r-\frac{1}{2}\diag(\v{\sigma}^{2})\right)(T - t)+\diag(\v{\sigma})\v{W}_{T}} \\
  &= \v{I} \; e^{\left(r-\frac{1}{2}\diag(\v{\sigma}^{2})\right)(T - t)+\diag(\v{\sigma})\v{W}_{T}}
  &&= \left( \begin{array}{cc} \frac{A_{1,T}}{a_1} & 0 \\ 0 & \frac{A_{2,T}}{a_2} \end{array} \right)
\end{align*}
from \eq{p_A_t} with suitably shifted time indices,
\eq{net_greeks_formal} simplifies to
\begin{align*}
  \frac{\partial}{\partial\v{a}}\v{x}_{t} & =\E_{t}^{Q}\left[
    e^{-r\tau}\left[\v{I}_{4\times4}-\frac{\partial}{\partial\v{x}}\v{g}(\v{a},\v{x})\right]^{-1}\left[\begin{array}{c}
\diag(\v{\xi})\\
\\
\diag(\v{1}-\v{\xi})
      \end{array}\right]
    \left( \begin{array}{cc} \frac{A_{1,T}}{a_1} & 0 \\ 0 & \frac{A_{2,T}}{a_2} \end{array} \right) \right]
\end{align*}
Furthermore, we find from \eq{g_diff_mat} that
\begin{align*}
  \v{I}_{4\times4}-\frac{\partial}{\partial\v{x}}\v{g}(\v{a},\v{x})
  &= \left( 
  \begin{array}{cccc}
    1 & - \xi_1 M^s_{12} & 0 & - \xi_1 M^d_{12} \\
    - \xi_2 M^s_{21} & 1 & - \xi_2 M^d_{21} & 0 \\
    0 & - (1-\xi_1) M^s_{12} & 1 & - (1-\xi_1) M^d_{12} \\
    - (1-\xi_2) M^s_{21} & 0 & - (1-\xi_2) M^d_{21} & 1
  \end{array} \right)
\end{align*}
Thus, the matrix is piecewise constant on each solvency region and
$\v{\Delta}$ can be found in all four cases. For illustration, we
detail the case $\xi_1 = \xi_2 = 1$, i.e.~both banks solvent:
{\footnotesize
\begin{align*}
  \left[ \v{I}_{4\times4}-\frac{\partial}{\partial\v{x}}\v{g}(\v{a},\v{x}) \right]^{-1}
  \left[\begin{array}{c} \diag(\v{\xi}) \\ \\ \diag(\v{1}-\v{\xi}) \end{array}\right]
  \left( \begin{array}{cc} \frac{A_{1,T}}{a_1} & 0 \\ 0 & \frac{A_{2,T}}{a_2} \end{array} \right)
  &= \left( 
  \begin{array}{cccc}
    1 & - M^s_{12} & 0 & - M^d_{12} \\
    - M^s_{21} & 1 & - M^d_{21} & 0 \\
    0 & 0 & 1 & 0 \\
    0 & 0 & 0 & 1
  \end{array} \right)^{-1}
  \left( \begin{array}{cc} 1 & 0 \\ 0 & 1 \\ 0 & 0 \\ 0 & 0 \end{array} \right)
  \left( \begin{array}{cc} \frac{A_{1,T}}{a_1} & 0 \\ 0 & \frac{A_{2,T}}{a_2} \end{array} \right) \\
  &= \left( 
  \begin{array}{cccc}
    \frac{1}{1 - M^s_{12} M^s_{22}} & \frac{M^s_{12}}{1 - M^s_{12} M^s_{22}} & 0 & 0 \\
    \frac{M^s_{21}}{1 - M^s_{12} M^s_{22}} & \frac{1}{1 - M^s_{12} M^s_{22}} & 0 & 0 \\
    0 & 0 & 1 & 0 \\
    0 & 0 & 0 & 1
  \end{array} \right)
  \left( \begin{array}{cc} 1 & 0 \\ 0 & 1 \\ 0 & 0 \\ 0 & 0 \end{array} \right)
  \left( \begin{array}{cc} \frac{A_{1,T}}{a_1} & 0 \\ 0 & \frac{A_{2,T}}{a_2} \end{array} \right) \\
  &= \left( 
  \begin{array}{cc}
    \frac{1}{1 - M^s_{12} M^s_{22}} \frac{A_{1,T}}{a_1} & \frac{M^s_{12}}{1 - M^s_{12} M^s_{22}} \frac{A_{2,T}}{a_2} \\
    \frac{M^s_{21}}{1 - M^s_{12} M^s_{22}} \frac{A_{1,T}}{a_1} & \frac{1}{1 - M^s_{12} M^s_{22}} \frac{A_{2,T}}{a_2} \\
    0 & 0 \\
    0 & 0 
  \end{array} \right)
\end{align*}}
containing the $\Delta$'s for equity (top) and debt (bottom)
respectively. Note that the equity $\Delta$'s correspond to the terms
for solvency region $\Rss$ in equations (\ref{eq:Delta11}) --
(\ref{eq:Delta22}).

Here, we have used that the inverse of a block matrix can be expressed
via the Schur complement as
\begin{align*}
  \left( \begin{array}{cc} \v{B}_{11} & \v{B}_{12} \\ \v{B}_{21} & \v{B}_{22} \end{array} \right)^{-1}
  &= \left( \begin{array}{cc} \v{I} & \v{0} \\ - \v{B}_{22}^{-1} \v{B}_{21} & \v{I} \end{array} \right)
  \left( \begin{array}{cc} ( \v{B}_{11} - \v{B}_{12} \v{B}_{22}^{-1} \v{B}_{21} )^{-1} & \v{0} \\ \v{0} & \v{B}_{22}^{-1} \end{array} \right)
  \left( \begin{array}{cc} \v{I} & - \v{B}_{12} \v{B}_{22}^{-1} \\ \v{0} & \v{I} \end{array} \right) \; .
\end{align*}
As in our case, $\v{B}_{11} = \v{I} - \v{M}^s, \v{B}_{12} = \v{M}^d,
\v{B}_{21} = \v{0}$ and $\v{B}_{22} = \v{I}$, we obtain
\begin{align*}
  \left( \begin{array}{cc} \v{I} & \v{0} \\ \v{0} & \v{I} \end{array} \right)
  \left( \begin{array}{cc} (\v{I} - \v{M}^s)^{-1} & \v{0} \\ \v{0} & \v{I} \end{array} \right)
  \left( \begin{array}{cc} \v{I} & - \v{M}^d \v{I} \\ \v{0} & \v{I} \end{array} \right)
  &= 
  \left( \begin{array}{cc} (\v{I} - \v{M}^s)^{-1} & \v{0} \\ \v{0} & \v{I} \end{array} \right)
\end{align*}
and the above results follows from
\begin{align*}
  \left( \begin{array}{cc} 1 & - M^s_{12} \\ - M^s_{21} & 1 \end{array} \right)^{-1}
  &= \frac{1}{1 - M^s_{12} M^s_{22}} \left( \begin{array}{cc} 1 & M^s_{12} \\ M^s_{21} & 1 \end{array} \right) \; .
\end{align*}

The $\Delta$'s on the other three solvency regions can be found
analogously and are ommitted for brevity.

\section{Additional figures}
\label{app:morefigs}

Markedly rising equity correlations at sufficiently low asset values
are also observed with additional equity cross-holdings of 10\%
(\fig{rhoS_10}) and asymmetric asset values (\fig{rhoS_2}). As the
figure is no longer symmetric with respect to the firms equity values,
we show the initial values of $a_{1}$ and $a_{2}$ (in log-scale)
instead.  The Suzuki areas, i.e., default boundaries, are indicated by
grey lines and the equity correlations are color coded. Again, at
sufficiently large debt cross-holding fractions a marked increase in
equity correlations is observed, especially in the $\Rdd$ Suzuki area,
i.e., during crisis.

\begin{figure}[ht]
  \centering
  \includegraphics[width=0.9\textwidth]{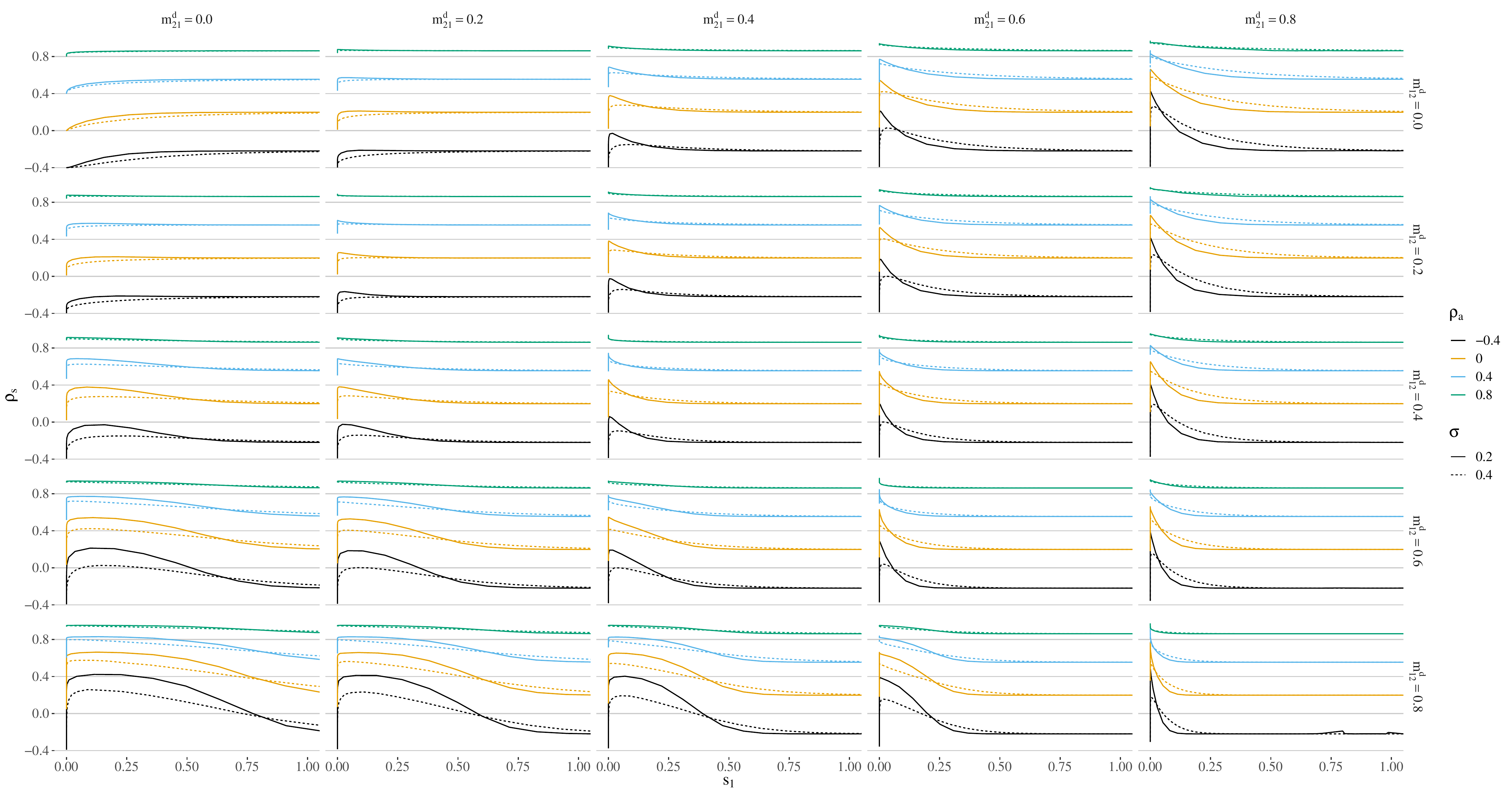}
  \caption{Same as \fig{rhoS_1}, but with additional equity
    cross-holdings of $M_{12}^{s} = M_{21}^{s} = 0.1$.
    \label{fig:rhoS_10} }
\end{figure}

\begin{figure}[ht]
\centering
\includegraphics[width=0.9\textwidth]{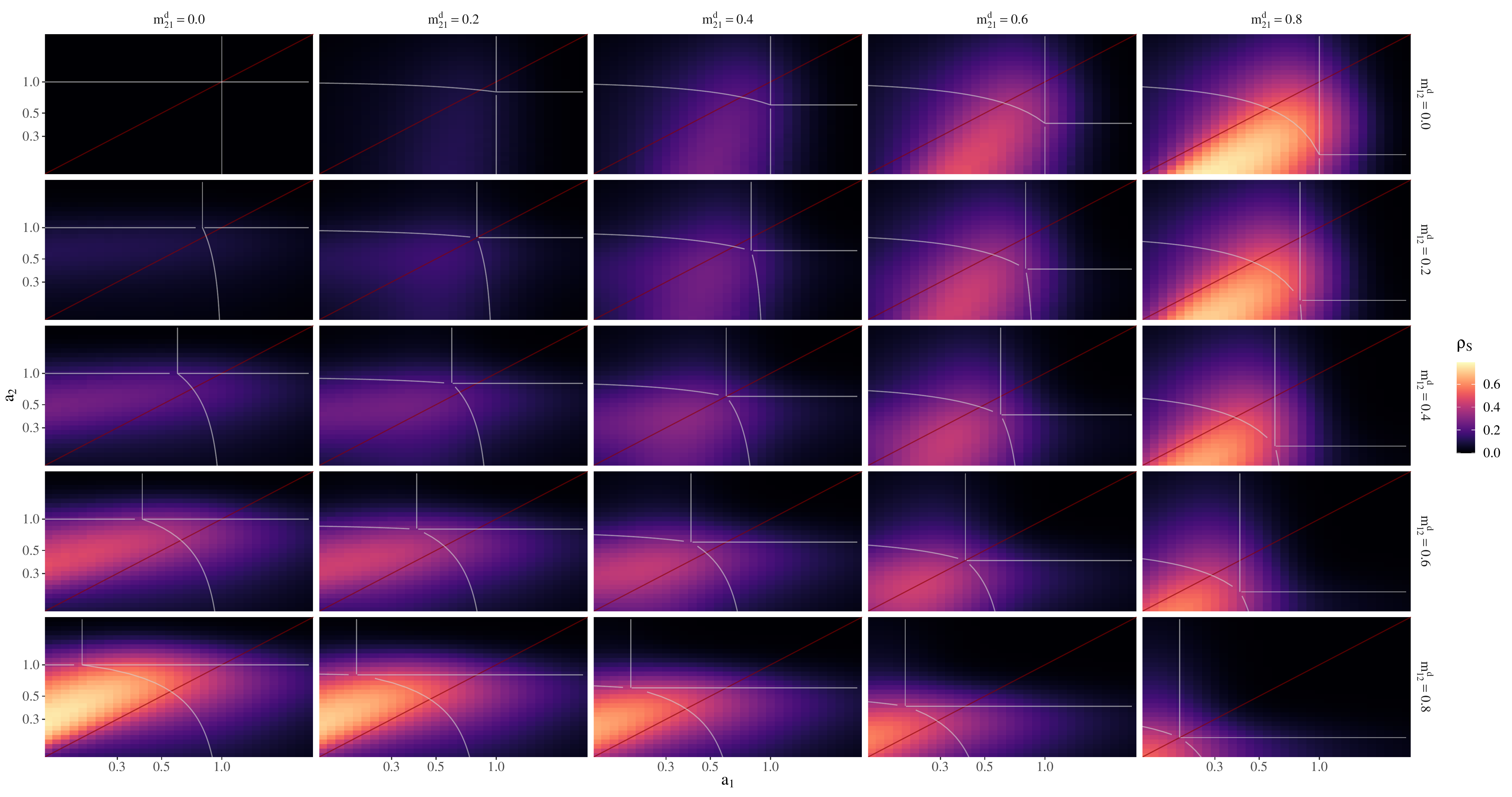} \caption{
  \label{fig:rhoS_2}
  Suzuki areas and equity correlations as a function of asset values
  $a_{1}$ and $a_{2}$. Here, the asset correlation is fixed at $\rho =
  0$.}
\end{figure}




\end{document}